\documentclass[a4paper,english]{lipics-v2021}
\hideLIPIcs
\nolinenumbers

\usepackage{amsthm}
\usepackage{amsmath}
\usepackage{amssymb}
\usepackage{amsfonts}
\usepackage{epsfig}
\usepackage{graphics}
\usepackage{graphicx}
\usepackage{booktabs}
\usepackage{pdfpages}
\usepackage{tikz}
\usepackage{pgfplots}
\usepackage[T1]{fontenc}

\usepackage{complexity}
\usepackage{url}

\usepackage{latexsym}
\usepackage{psfrag}
\usepackage{enumerate}
\usepackage{here}
\usepackage{esvect}
\usepackage[update,prepend]{epstopdf}
\usepackage{anyfontsize}
\usepackage{url}
\usepackage{hyperref}

\usepackage{microtype}

\newcommand{\ie}{{i.e.}}

\newcommand{\ZZ}{\mathbb{Z}} 
\newcommand{\RR}{\mathbb{R}} 

\newcommand*\circled[1]{\tikz[baseline=(char.base)]{
    \node[shape=circle,draw,inner sep=2pt] (char) {#1};}}

\newcommand{\und}{\underline}

\def\Prob{{\rm Prob}}

\newcommand{\e}{\mathrm{e}}
\newcommand{\x}{\mathrm{x}}

\usepackage[noline,linesnumbered]{algorithm2e}
\SetArgSty{textrm}
\let\oldnl\nl
\newcommand{\nonl}{\renewcommand{\nl}{\let\nl\oldnl}}
\makeatletter
\def\TitleOfAlgo{\@ifnextchar({\@TitleOfAlgoAndComment}{\@TitleOfAlgoNoComment}}
\def\@TitleOfAlgoAndComment(#1)#2{\nonl\hspace*{-1.5em}#2 #1\;}
\def\@TitleOfAlgoNoComment#1{\nonl\hspace*{-1.5em}#1\;}
\makeatother
\DontPrintSemicolon

\newcommand{\later}[1]{}
\newcommand{\old}[1]{}

\title{A Stopping Game on Zero-Sum Sequences}

\titlerunning{A Stopping Game on Zero-Sum Sequences}

\author{Adrian Dumitrescu}
{Algoresearch L.L.C., Milwaukee, WI, USA}
{ad.dumitrescu@algoresearch.org}
{0000-0002-1118-0321}
{}
\author{Arsenii Sagdeev}
       {Karlsruhe Institute of Technology, Karlsruhe, Germany,\\
         Alfr\'ed R\'enyi Institute of Mathematics, Budapest, Hungary
       }
{sagdeevarsenii@gmail.com}
{}
{}

\authorrunning{Adrian Dumitrescu and Arsenii Sagdeev}

\Copyright{Adrian Dumitrescu and Arsenii Sagdeev}

\funding{Work on this paper was partially supported by ERC Advanced Grant no.\! 882971 ``GeoScape".}



\acknowledgements{The second author thanks Fedor Petrov for a helpful discussion of Lemma 3.}

\begin{document}

\maketitle

\begin{abstract}
 We introduce and analyze a natural game formulated as follows.
 In this one-person game, the player is given a random permutation $A=(a_1,\dots, a_n)$ of a multiset $M$ of $n$ reals
 that sum up to $0$, where each of the $n!$ permutation sequences is equally likely.
The player only knows the value of $n$ beforehand.
The elements of the sequence are revealed one by one and the player can stop the game at any time.
 Once the process stops, say, after the $i$th element is revealed, the player collects the amount
 $\sum_{j=i+1}^{n} a_j$ as his/her payoff and the game is over (the payoff corresponds to the
 unrevealed part of the sequence). 

Three online algorithms are given for maximizing the expected payoff in the binary case when $M$ contains only $1$'s and $-1$'s.
{\sc Algorithm~1} is slightly suboptimal, but is easier to analyze. 
Moreover, it can also be used when $n$ is only known with some approximation.
{\sc Algorithm~2} is exactly optimal but not so easy to analyze on its own.
{\sc Algorithm~3} is the simplest of all three. 
It turns out that the expected payoffs of the player are $\Theta(\sqrt{n})$ for all three algorithms. 

In the end, we address the general problem and deal with an arbitrary zero-sum multiset,
for which we show that our {\sc Algorithm 3} returns a payoff proportional to $\sqrt{n}$,
which is worst case-optimal. 
\end{abstract}

\section{Introduction} \label{sec:intro}

\emph{Stopping rules} are of key importance in the study of random processes and/or
online algorithms. Here is an illustrative example concerning the behavior of the largest payoff expectation
in a game proposed by Leo Moser almost $70$ years ago~\cite{Mo56}.

A real number is drawn at random from $[0,1]$. We may either keep the number selected, or
reject it and draw again. We can then either keep the second number chosen or reject it, and draw again,
and so on. Suppose we have at most $n$ choices, and the payoff is the number we keep.
What stopping rule gives the largest payoff expectation $E_n$ and how can we estimate $E_n$?

Moser showed that, for any distribution, given at most $n$ choices,
one should stop after the first choice if and only if the outcome is at least $E_{n-1}$.
In particular, for the uniform distribution, $f_1(x)=1$ for $0 \leq x \leq 1$,
$E_n$ satisfies the  recurrence
$E_{n+1} = \frac12 (1 + E_n^2)$, with $E_0=0$, $E_1=1/2$. 
For illustration, we list some values in Table~\ref{tab:E_n}.
Observe that $\lim_{n \to \infty} E_n =1$. 
The case of the normal distribution was studied by Guttman~\cite{Gu60}.
Other examples of stopping rules, mostly from the area of selection, are mentioned in the ``Related work''
paragraph later in this section. See also~\cite[Chap.~13]{MU17}.

\begin{table}[ht]
\caption{$E_n$ in Moser's procedure for the uniform distribution (the first three digits).}
\begin{center}
\begin{tabular}{|c|c||c|c||c|c|}
\hline
$n$ & $E_n$ & $n$ & $E_n$ & $n$ & $E_n$ \\
\hline
$1$ & $.500$ & $5$ & $.775$ & $100$ & $.981$ \\
$2$ & $.625$ & $10$ & $.861$ &  $200$ & $.990$ \\
$3$ & $.695$ & $20$ & $.919$ & $300$ & $.993$ \\
\hline
\end{tabular}
\end{center}
\label{tab:E_n}
\end{table}

Here we consider the following one-person game. Let $M$ be a (multi)set of $n$ reals that sum up to $0$.
The value of $n$ is known to the player beforehand. The player is given a random permutation $A=(a_1,\dots,a_n)$ of $M$,
where each of the $n!$ permutation sequences is equally likely. The elements of the sequence are revealed one by one and the player
can stop the game at any step, \ie, after the $a_i$ is revealed,
 for any $i=1,2,\ldots,n$, or the stop is automatic at the end of the process.
 Once the process stops, say, after the $i$th element is revealed, the player collects the amount
 $\sum_{j=i+1}^{n} a_j$ as his/her payoff and the game is over. Note that the payoff corresponds to the
 unrevealed part of the sequence. 

 Observe that any rational strategy for the game achieves a nonnegative outcome,
 since the player can always reduce its losses by continuing to play until the end,
 where he/she could force a zero payoff.

 In the \textit{binary case}, i.e., if $n=2m$ and $M$ consists of $m$ copies of $-1$'s and $1$'s, respectively,
 we define the \emph{state} of the process as the pair $(i,j)$, where $i$ and $j$ denote
the numbers of $-1$ and $+1$, revealed so far, respectively.

\subparagraph{Example.} Let $A=(-1,1,1,1,-1,1,-1,-1)$, with $n=8$. 
If the player stops the game before the first element (or after the last one) is revealed, his/her payoff is zero.
If the player stops the game before the second element is revealed, his/her payoff is $1$ (the total
amount of the unrevealed subsequence). Stopping the game at any other step would be irrational --- and is
therefore omitted from discussion.

\subparagraph{Comment.} The dual version of the game in which the payoff is the sum of the revealed values
up to the stopping point admits a similar analysis.
That is,  if the player stops the game after the $i$th element of $A$ is revealed,
his/her payoff is $\sum_{j=1}^{i} a_j$ (rather than $\sum_{j=i+1}^{n} a_j$).

\subparagraph{Our results.}
We present three online algorithms for maximizing the expected payoff in the binary case.
{\sc Algorithm~1} (in Section~\ref{sec:simple}) is slightly suboptimal,
but is easier to analyze. Moreover, it can also be used when $n$ is only known
with some approximation (in Section~\ref{sec:remarks}).

{\sc Algorithm~2} (in Section~\ref{sec:optimal}) is optimal but not so easy to analyze
on its own. This algorithm requires quadratic preprocessing time but is otherwise very simple to use.

{\sc Algorithm~3} (in Sections~\ref{sec:middle-b} and~\ref{sec:middle-g})  is the simplest of all three.
Moreover, it is applicable to any zero-sum multiset.

We show that the three algorithms share the same asymptotic performance, 
$\Theta(\sqrt{n})$, which is worst-case optimal. Our results are summarized in the following.

\begin{theorem} \label{thm:main}
Let $M=\{a_1,a_2,\ldots,a_n\}$ be an arbitrary zero-sum multiset of reals. 
There is an online algorithm that reads a random permutation of $M$ and stops with an expected payoff
of $\Omega(\mu \sqrt{n})$, where $\mu =\sum_{i=1}^n |a_i|/n$ is the average absolute value
for the input sequence. This bound is worst-case optimal, \ie, equality is attained 
(apart from the hidden constant factor) for a uniform sequence, with $n/2$ values equal to 
$\mu$ and $-\mu$, respectively. 
\end{theorem}

\subparagraph{Related work.}
The \emph{best choice} problem for partially ordered sets considers a selection process
in an online setting~\cite{GKMN08,Pr99}. 
Let $X$ be a finite set equipped with a strict partial order $\prec$, \ie, an asymmetric and transitive relation.
The objects (elements) of $X$ are shuffled  --- each of the $|X|!$ permutations is equally likely ---
and appear one by one before a selector $\S$, who must make a single online choice. 
The $t$th element $x_t$ appears at time $t=1,2\ldots,n$.
The selector knows the cardinality $n=|X|$ but is \emph{ignorant} of $\prec$. 
$\S$ has access to the comparisons between elements that appeared so far;
thus after time $t$, $\S$ only knows the relation $\{(i,j) \colon i,j \leq t, x_i \prec x_j \}$.

The selector's goal is to pick an element that is \emph{maximal} with respect to $\prec$. 
That is, his/her goal is to stop the process at some time $\tau$
(that depends on the permutation) so that the currently examined element $x_t$ is \emph{maximal}
with respect to $\prec$. Such a random variable $\tau$ is called a \emph{stopping time}. 
The question that arises is: (Q) Can the selector obtain a maximal element with probability at least
$\delta$, for some constant $\delta$? --- irrespective of $(X,\prec)$. 
Preater~\cite{Pr99} was the first to obtain such an algorithm: the selector can pick a maximal
element with probability at least $1/8$. The authors of~\cite{GKMN08} improved the analysis
of Preater's algorithm and showed that the success probability is at least $1/4$, which is tight
for that algorithm. 

The case when $\prec$ is a total (\ie, linear) order, and thus known to $\S$, is the classic
\emph{secretary problem}~\cite{Fe89,LM85}. In the language of online algorithms,
one can refer to the problem as the \emph{online maximum} problem: the elements in a random
sequence $A$ of $n$ (distinct) numbers from an unknown distribution are revealed one by one,
and the goal is to select the maximum by an irrevocable decision made immediately after examining 
the $i$th element for some $i=1,2,\ldots,n$. Each of the $n!$ permutations of the ranks is equally likely.
Note that inaction leads to the automatic, \ie, forced, selection of the last (the $n$th) element
in the sequence.

The optimal algorithm achieves this task with a probability of success bounded from below by a
constant, see~\cite[pp.~73--77]{Mo65} and~\cite[pp.~90--95]{LM85}. For large $n$, the optimal strategy
passes approximately a fraction of $1/e$ of the sequence and chooses the first element
that is larger every element of the initial prefix. Then the
probability of success is approximately $1/e=0.367 \ldots$.
Similar algorithms, whose probabilities of success are bounded from below by a constant, 
can be designed for finding the \emph{online minimum}, \emph{online minimum and maximum},
\emph{online smallest two and largest two}, and for other variants. 

Under the additional assumption that the random numbers are uniformly distributed between $0$ and $1$,
namely in Moser's problem described in the introduction, the probability of success of the optimal strategy
is about $0.580$ (for large $n$); see~\cite[pp.~77--79]{Mo65},~\cite{Fe89}.

The success of an algorithm for optimal selection can be also quantified by the rank of the selected element,
where, say, the best rank is $1$ and the worst rank is $n$. Sixty years ago the authors of~\cite{CMRS64}
obtained the elegant result that for the stopping rule which minimizes the expectation of the absolute
rank of the selected element, as $n \to \infty$, the expected rank tends to the value of the infinite product
$ \prod_{j=1}^\infty \left( \frac{j+2}{j} \right)^{\frac{1}{j+1}} \approx 3.8695$.

\section{Preliminaries} \label{sec:prelim}

A \emph{lattice path} in $\ZZ^2$ is a path that moves right or up by one unit at each step.
Consider the integer lattice points in the square $[0,m] \times [0,m]$ and the set $L_m$ of 
lattice paths connecting the lower left corner of the square $(0,0)$ with the upper
right corner of the square $(m,m)$.

In the binary case, the evolution of the game corresponds to a prefix of a path in $L_m$. 
The path moves a step to the right if a revealed element is $-1$, and moves a step up if it is $+1$.
If the player stops the game after the $i$th element is revealed ($i \geq 1$), the path length is $i$.
See Fig.~\ref{fig:paths}.

\begin{figure}[htbp]
\centering
\begin{tikzpicture}[scale=0.55]
\draw[step=1cm,gray,very thin] (-1.9,-1.9) grid (7.9,7.9);
\draw[thick,->] (0,0) -- (8.5,0);
\draw[thick,->] (0,0) -- (0,8.5);
\foreach \x in {0,1,2,3,4,5,6,7}
   \draw (\x cm,1pt) -- (\x cm,-1pt) node[anchor=north] {$\x$};
\foreach \y in {0,1,2,3,4,5,6,7}
\draw (1pt,\y cm) -- (-1pt,\y cm) node[anchor=east] {$\y$};
\draw[line width=2pt] (0,0) -- (0,1) -- (4,1) -- (4,5) -- (5,5) -- (5,6) -- (6,6);
\draw[blue] [line width=1pt] (0,0) -- (7,7);
\draw[red] [line width=1pt] (0,-2) -- (7,5);
\end{tikzpicture}
\hspace{10mm}
\begin{tikzpicture}[scale=0.55]
\draw[step=1cm,gray,very thin] (-1.9,-1.9) grid (7.9,7.9);
\draw[thick,->] (0,0) -- (8.5,0);
\draw[thick,->] (0,0) -- (0,8.5);
\foreach \x in {0,1,2,3,4,5,6,7}
   \draw (\x cm,1pt) -- (\x cm,-1pt) node[anchor=north] {$\x$};
\foreach \y in {0,1,2,3,4,5,6,7}
\draw (1pt,\y cm) -- (-1pt,\y cm) node[anchor=east] {$\y$};
\draw[line width=2pt] (0,0) -- (0,2) -- (3,2) -- (3,5) -- (6,5) -- (6,6);
\draw[blue] [line width=1pt] (0,0) -- (7,7);
\draw[red] [line width=1pt] (0,-2) -- (7,5);
\end{tikzpicture}
\caption{Left: a lattice path from $(x,y)=(0,0)$ to $(6,6)$ that reaches the red line $y=x-2$ (left), 
  and one that does not (right). Here $n=6$. The line $y=x$ is in blue.}
\label{fig:paths}
\end{figure}
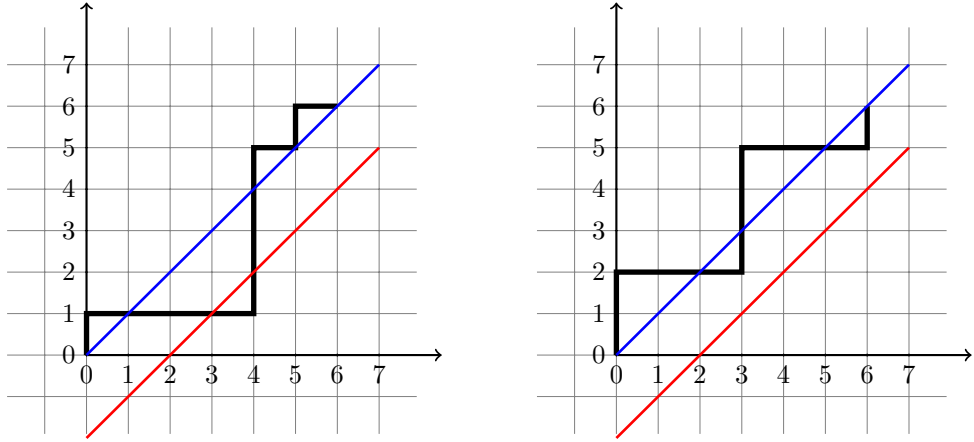

Recall the following result regarding lattice paths  --- obtained by the so-called
\emph{reflection principle}:

\begin{lemma} \label{lem:reach} {\rm \cite[Ch.~1.3]{Mo79}}
Let $t \geq 1$ be a positive integer.
The number of lattice paths from $(0,0)$ to $(m,m)$ that reach line $\ell: y=x-t$ is
${2m \choose m-t}$. 
\end{lemma}

Let $E_t$ be the event that a random lattice path in $L_m$ reaches line $\ell: y=x-t$.
Recall that the total number of paths in $L_m$ is $|L_m|={2m \choose m}$.
Then by Lemma~\ref{lem:reach} we obtain:
\begin{equation} \label{eq:reach}
  \Prob\left(E_t\right) = \frac{{2m \choose m-t}}{{2m \choose m}}.
\end{equation}

\section{A simple threshold algorithm for the binary case} \label{sec:simple}

\subparagraph{Algorithm outline.} An integer threshold value $t =\Theta(\sqrt{n})$ is set; where $n=2m$. 
The algorithm waits until the lattice path reaches the line $\ell: y=x-t$, if at all, when it stops the game. 
If the lattice path describing the evolution of the game doesn't reach the line $\ell$, the algorithm continues
until the end (\ie, the entire input sequence is revealed). Next we proceed with the details.

\smallskip
\begin{algorithm}[H]
  \DontPrintSemicolon
  \TitleOfAlgo{\textsc{Algorithm 1}}
  \KwIn{the elements of a zero-sum sequence $A$ of length $2m$ over $\Sigma=\{-1,+1\}$ are revealed one by one}
  Set $t = \left \lfloor \sqrt{\frac{m+1}{2}} \right \rfloor$\;
  Follow the lattice path from $(0,0)$ to $(m,m)$ ($\rightarrow$ or $\uparrow$) as dictated
  by the entries of $A$ that have been revealed (step by step) by repeating Step 3\;
  If the path has reached the line $\ell: y=x-t$ \underline{stop}; else \underline{continue}\;
\end{algorithm}
\smallskip

\subparagraph{Algorithm analysis.} We first deduce a lower bound on the expected payoff.
Set $t = \left \lfloor \sqrt{\frac{m+1}{2}} \right \rfloor$. 
From~\eqref{eq:reach} we further obtain by 
applying the standard inequality $1-x \geq \e^{-2x}$ for $0 \leq x \leq 1/2$ in the last step:
\begin{align}
\Prob\left(E_t\right) &= \frac{{2m \choose m-t}}{{2m \choose m}} 
  = \frac{m(m-1) \cdots (m-t+1)}{(m+t)(m+t-1) \cdots (m+1)} \nonumber \\
  &= \prod_{i=1}^{t} \left( 1 - \frac{t}{m+i} \right)  \label{eq:prod1} \\
  &\geq \left(1 - \frac{t}{m+1} \right)^t \geq \exp \left(\frac{-2t^2}{m+1} \right) \label{eq:prod2} \\
  &\geq \exp(-1).  \nonumber
\end{align}

Let $W_1$ denote the expected payoff taken by the player. For the above value of $t$ we have:
\begin{equation}  \label{eq:prod3}
W_1 = t \cdot \Prob\left(E_t\right) \geq
\frac{1}{\e} \left \lfloor \sqrt{\frac{m+1}{2}} \right \rfloor = \Omega(\sqrt{m}).
\end{equation}

As for the upper bound, observe that the payoff is at least $t$ with probability only at most $\Prob\left(E_t\right)$
even if the player has complete information about a given random sequence, and chooses the best moment to stop the game.
Hence, the expected payoff of \textit{any} algorithm is at most the sum of these probabilities (see~\cite[Lemma~2.9]{MU17}):
\begin{equation*}
  \sum_{t=1}^{m}\Prob\left(E_t\right) = \sum_{t=1}^{m}\frac{{2m \choose m-t}}{{2m \choose m}}
  = \frac{2^{2m-1}}{{2m \choose m}}-\frac{1}{2} = \frac{\sqrt{\pi m }}{2}-\frac{1}{2}+O(m^{-1/2}),
\end{equation*}
where the last equality is a routine series expansion. Hence, $W_1 \!=\! \Theta(\sqrt{m})$, as claimed.

Executing one step of \textsc{Algorithm 1} takes constant time.

\subparagraph{Remarks.}
Our findings apply with some adjustments to the following variants:
\begin{itemize} \itemsep 3pt
  
\item The player is required to stop the game before the last element in the sequence is revealed. 
  Since the state of the process is known at any step, the player can always reduce its losses, if any
  are perceived at the current step, by continuing to play until exactly one element, the last one,
  remains, where he/she can stop the game and so the loss can be limited to $-1$. The analysis of this
  variant is similar.
  
\item The length $n$ of the input sequence is only known approximately, say, $N  \leq n \leq 2N$,
  for a given $N$.   For this variant, the player can use the {\sc Algorithm 1} with a similar 
  expected performance, where $n$ is replaced by $N$. 

\end{itemize}

\section{An optimal algorithm for the binary case} \label{sec:optimal}

The algorithm constructs two square matrices, whose content is described below.

Let $\left[ T[i,j] \right]_{0 \leq i,j \leq m}$ be a $(m+1)\times (m+1)$ matrix with nonnegative entries,
where $T[i,j]$ is the expected value of the player's payoff if he/she starts playing optimally
in state $(i,j)$, \ie, after $i$ elements of $-1$ and $j$ elements of $+1$ have been
revealed. As such, $T[0,0]$ is the expected \emph{value} of the game in an optimal play.

The optimal strategy is encoded in the \emph{stopping} square matrix $S$ as described below.
Let $\left[ S[i,j] \right]_{0 \leq i,j \leq m}$ be a $(m+1)\times (m+1)$ matrix with 0-1 entries,
where
\[S[i,j] = \begin{cases}
    1 & \text{if stop at this step --- in state } (i,j),\\
    0 &\text{if continue at this step --- in state } (i,j).
\end{cases}
\]

The matrix $T$ is constructed in a \emph{backwards} manner by using dynamic programming.
Indeed, the last row and column of $T$ are known beforehand: $T[i,m]=0$, for  $i=0,1,\ldots,m$,
and $T[m,j]=m-j$ for $j=0,1,\ldots,m$. Assuming that the entries
$T[i+1,j]$ and $T[i,j+1]$ are known, the entry $T[i,j]$ can be computed as follows:

First, the transition probabilities $p_0$ and $p_1$ to reach states $(i+1,j)$ and $(i,j+1)$, respectively,
by revealing a new element in state $(i,j)$ are computed:
\begin{equation*}  \label{eq:p0p1}
	p_0 = \frac{m-i}{2m-i-j}, ~~~
	p_1 = \frac{m-j}{2m-i-j}.
\end{equation*}

Then
\begin{equation} \label{eq:max}
T[i,j] = \max \left( i-j, p_0 \cdot T[i+1,j] + p_1 \cdot T[i,j+1] \right). 
\end{equation}

As such, the entries of $T$ can be computed row by row, from top to bottom,
and each row from right to left.

Consider an entry $T[i,j]$; If the maximum in~\eqref{eq:max} is given by the first term, $i-j$,
the corresponding entry in $S$ is set to $1$: $S[i,j]=1$. Indeed, this means that
the  expected payoff if the algorithm continues from state $(i,j)$ does not exceed the
payoff available in state $(i,j)$ if the algorithm stops. Otherwise, if the maximum in~\eqref{eq:max}
is given by the second term, set $S[i,j]=0$. Some of the stopping entries are unreachable
(for example, in the last column), but this is not an issue. 

In summary, the algorithm works as follows:

\smallskip
\begin{algorithm}[H]
  \DontPrintSemicolon
  \TitleOfAlgo{\textsc{Algorithm 2}}
  \KwIn{the elements of a zero-sum sequence $A$ of length $2m$ over $\Sigma=\{-1,+1\}$ are revealed one by one}
  Compute the matrices $T$ and $S$ as specified above\;
  Follow the path in $S$ starting from the entry $(0,0)$ by repeating Step 3\;
  If $S[i,j]=1$ (\ie, if $(i,j)$ is a stopping state), \underline{stop};
  else \underline{continue} to the next entry in $S$ ($S[i+1,j]$ or $S[i,j+1]$)
  as dictated by the most recently revealed element of~$A$\;
\end{algorithm}
\smallskip

\smallskip
Two examples with the optimal strategy as computed by \textsc{Algorithm 2} appear in
Fig.~\ref{fig:optimal}. 

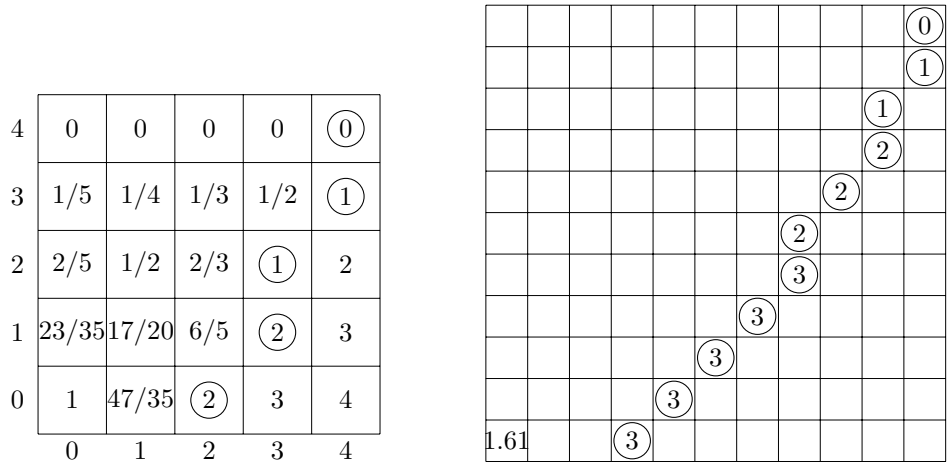
\begin{figure}[htbp]
 \centering 
\begin{tikzpicture} [scale=0.9]
  \draw[step=1cm,black,very thin] (0,0) grid (5,5);
\draw (0.5,-0.25) node{$0$};
\draw (1.5,-0.25) node{$1$};
\draw (2.5,-0.25) node{$2$};
\draw (3.5,-0.25) node{$3$};
\draw (4.5,-0.25) node{$4$};
\draw (-0.3,0.5) node{$0$};
\draw (-0.3,1.5) node{$1$};
\draw (-0.3,2.5) node{$2$};
\draw (-0.3,3.5) node{$3$};
\draw (-0.3,4.5) node{$4$};
\draw (0.5,0.5) node{$1$};
\draw (1.5,0.5) node{$47/35$};
\draw (2.5,0.5) node{$\circled{2}$};
\draw (3.5,0.5) node{$3$};
\draw (4.5,0.5) node{$4$};
\draw (0.5,1.5) node{$23/35$};
\draw (1.5,1.5) node{$17/20$};
\draw (2.5,1.5) node{$6/5$};
\draw (3.5,1.5) node{$\circled{2}$};
\draw (4.5,1.5) node{$3$};
\draw (0.5,2.5) node{$2/5$};
\draw (1.5,2.5) node{$1/2$};
\draw (2.5,2.5) node{$2/3$};
\draw (3.5,2.5) node{$\circled{1}$};
\draw (4.5,2.5) node{$2$};
\draw (0.5,3.5) node{$1/5$};
\draw (1.5,3.5) node{$1/4$};
\draw (2.5,3.5) node{$1/3$};
\draw (3.5,3.5) node{$1/2$};
\draw (4.5,3.5) node{$\circled{1}$};
\draw (0.5,4.5) node{$0$};
\draw (1.5,4.5) node{$0$};
\draw (2.5,4.5) node{$0$};
\draw (3.5,4.5) node{$0$};
\draw (4.5,4.5) node{$\circled{0}$};
\end{tikzpicture}
\hspace{10mm}
\begin{tikzpicture}[scale=0.55]
\draw[step=1cm,black,very thin] (0,0) grid (11,11);
\draw (0.5,0.5) node{$1.61$};
\draw (3.5,0.5) node{$\circled{3}$};
\draw (4.5,1.5) node{$\circled{3}$};
\draw (5.5,2.5) node{$\circled{3}$};
\draw (6.5,3.5) node{$\circled{3}$};
\draw (7.5,4.5) node{$\circled{3}$};
\draw (7.5,5.5) node{$\circled{2}$};
\draw (8.5,6.5) node{$\circled{2}$};
\draw (9.5,7.5) node{$\circled{2}$};
\draw (9.5,8.5) node{$\circled{1}$};
\draw (10.5,9.5) node{$\circled{1}$};
\draw (10.5,10.5) node{$\circled{0}$};
\end{tikzpicture}
 \caption{Left: an optimal strategy for $n=8$: $T[0,0]=1$,  $T[1,0]=47/35$, $T[2,0]=2$, $\ldots$, $T[4,4]=0$.  
   Right: an optimal strategy for $n=20$; to avoid the clutter, only the reachable stopping states are shown.
   Note that all stopping states are below the main diagonal (with one exception, the zero entry for the
   automatic stop in the end).}
	\label{fig:optimal}
\end{figure}

\subparagraph{Algorithm analysis.}
For any given state $(i,j)$ the algorithm has two options: stop or continue.
We may assume w.l.o.g. that $i \geq j$.
Note that payoff of value $i-j$ is immediately available for collection if the
algorithm stops in state $(i,j)$. 
Recall that each zero-sum sequence is equally likely and thus, for a given prefix, every
suffix of it is equally likely. 
Note that the second term in~\eqref{eq:max} is equal to the expected payoff
if the algorithm continues from state $(i,j)$.

If the maximum in~\eqref{eq:max} is given by the first term, $i-j$, this means that
the  expected payoff if the algorithm continues from state $(i,j)$ does not exceed the
payoff available in state $(i,j)$ if the algorithm stops. Recall that $S[i,j]=1$ in this case and thus
the algorithm correctly stops in this state. 
If the maximum in~\eqref{eq:max} is given by the second term, this means that
the  expected payoff if the algorithm continues from state $(i,j)$ exceeds those
available in state $(i,j)$ if the algorithm stops. Recall that $S[i,j]=0$ in this case and thus
the algorithm correctly continues from this state. 

Note that both $T$ and $S$ can be computed in quadratic time.  
Executing one step of \textsc{Algorithm~2} takes constant time. 

Let $W_2$ be the expected payoffs of the player under {\sc Algorithm 2}. Since {\sc Algorithm 2} is optimal,
it is stronger than {\sc Algorithm 1} from Section~\ref{sec:simple}, and thus $W_2 \geq W_1 =\Omega(\sqrt{n})$.
In addition, the general upper bound analysis in Section~\ref{sec:simple} yields that every algorithm,
including {\sc Algorithm 2}, has a payoff at most $O(\sqrt{n})$. Consequently, $W_2 = \Theta(\sqrt{n})$,
as claimed. 

\subparagraph{Example.}  With a shuffled standard deck of $52$ cards ($n=52$), the player
using \textsc{Algorithm~2} is guaranteed \$2.62 in expectation (which is optimal), if the payoff is,
say, the number of black cards minus the number of red cards in what is left unturned from the deck,
with \$1 per card. For comparison, the player using \textsc{Algorithm~1} is guaranteed
\$1.54 in expectation, by~\eqref{eq:prod3} (however, no optimization was attempted in its analysis).

\section{A `Stop in the Middle' algorithm for the binary case} \label{sec:middle-b}

In this section we offer yet another algorithm for the binary case, \textsc{Algorithm 3}.
Its expected payoff is $\Omega(\sqrt{n})$, asymptotically the same as for the previous two algorithms.
This algorithm is extremely simple as it does not perform any bookkeeping. However,
the main reason for presenting it is its applicability to arbitrary zero-sum sequences
as shown in Section~\ref{sec:middle-g}. 

For convenience, here we formulate the algorithm for the dual version of the game, 
in which,  if the player stops the game after the $i$th element of $A$ is revealed,
his/her payoff is $\sum_{j=1}^{i} a_j$.

Consider the function $g \colon \RR \to \RR$ defined as
 \begin{equation}
    g(x) =
    \begin{cases}
      x & \text{ if } x>0 \\
      0       & \text{otherwise.}
    \end{cases}
 \end{equation}
 The algorithm's payoff is $g(\sum_{j=1}^{n/2} a_j)$. Note that this only requires computing a running sum
 up to index $n/2$ and making a single sign determination.

\smallskip
\begin{algorithm}[H]
  \DontPrintSemicolon
  \TitleOfAlgo{\textsc{Algorithm 3}}
  \KwIn{the elements of a zero-sum sequence $A$ of length $n$ are revealed one by one}
  Compute the running sum $S=\sum_{j=1}^{n/2} a_j$ step by step\;
  If $S>0$ \underline{stop}; else \underline{continue} to the end \;
\end{algorithm}
\smallskip

Assume for convenience in the analysis that $n=4m$. Then it is easy to see that in the binary case ($\Sigma=\{-1,+1\}$),
the total number of zero-sum sequences is ${4m \choose 2m}$, whereas there are exactly ${2m \choose m-k}{2m \choose m+k}$
of them for which the payoff is $2k$, $0 \le k \le m$. Hence, the expected payoff $W_3(n)$ is
\begin{equation} \label{eq:alg3}
W_3(n) = {4m \choose 2m}^{-1} \sum_{k=1}^m 2k {2m \choose m-k}{2m \choose m+k}.
\end{equation}
A few values of $W_3(n)$ for some small even $n$ are listed in Table 2.

\begin{table}[ht]
\caption{A few values of $W_3(n)$.}
\begin{center}
  \begin{tabular}{|c|c|c|c|c|c|c|c|c|c|c|} 
    \hline
$n$ & 2 & 4 & 6 & 8 & 10 & 12 & 14 & 16 & 32 & 64\\
    \hline
    $W_3(n)$ & $\frac12$ & $\frac13$ & $\frac35$ & $\frac{18}{35}$ & $\frac57$ &
    $\frac{50}{77}$ & $\frac{350}{429}$ & $\frac{980}{1287}$ & $\approx 1.10$ & $\approx 1.57$ \\
    \hline
\end{tabular}
\end{center}
\label{tab:W_3}
\end{table}


\subparagraph{Algorithm analysis.} The following standard identity, which we prove in Appendix~\ref{sec:app},
leads to a closed form expression of the expected payoff.

\begin{lemma} \label{lem:W3} 
\begin{equation*}
	\sum_{k=1}^m 2k {2m \choose m-k}{2m \choose m+k} = 2m{2m-1 \choose m-1}^2.
\end{equation*}
\end{lemma}

Now a routine series expansion yields
\[ W_3(n) = 2m{2m-1 \choose m-1}^2 {4m \choose 2m}^{-1} =
\sqrt{\frac{m}{2\pi}} + O(m^{-1/2}) = \sqrt{\frac{n}{8\pi}} + O(n^{-1/2}). \]
Therefore, $W_3(n) =\Theta(\sqrt{n})$, asymptotically the same as for the previous two algorithms.

\section{The `Stop in the Middle' algorithm for the general case} \label{sec:middle-g}

In this section we generalize our results to an \emph{arbitrary} zero-sum multiset $M=\{a_1,\ldots,a_n\}$
and thereby prove Theorem~\ref{thm:main}. 
Assume for convenience that the input size $n$ is even (the treatment for the case of odd $n$ is similar,
but somewhat more involved, and is therefore omitted). 
As in Section~\ref{sec:middle-b}, here we consider the dual version of the game, 
in which the payoff is the sum of the revealed elements. 

We show that the expected payoff achieved by \textsc{Algorithm 3} on a random permutation
of $M$ is $f(M) =\Omega(\mu \sqrt{n})$, where $\mu =\sum_{i=1}^n |a_i|/n$ is the
\emph{average absolute value} for the input multiset. Let $n_+$ and $n_-$ denote the number of positive and respectively
negative entries in $M$; where $n=n_+ + n_-$. Let $m=n/2$. Note that $f(\cdot)$ is a linear function,
\ie, $f(\lambda M) = \lambda f(M)$,
where $\lambda M = \lambda \{a_1,\ldots,a_n\} = \{\lambda a_1, \ldots, \lambda a_n\}$.

The key lemma below shows that the expected payoff can only decrease when replacing two distinct elements
by two copies of their average. If the two elements are both positive or both negative, this type of replacement
maintains $n_+$ and $n_-$.  If the two elements have different signs, the replacement can be used
to adjust the counts $n_+$ and $n_-$, as illustrated after the proof of the lemma. 

\begin{lemma} \label{lem:replace} 
  Let $a,b \in M$, $a<b$, be two distinct elements of $M$ and $M'= M \setminus \{a,b\}$. Then
  \[ f\left(M' \cup \{a,b\}\right) \geq f\left(M' \cup \left\{\frac{a+b}{2}, \frac{a+b}{2}\right\}\right). \]
\end{lemma}
\begin{proof}
Let $F \subset M$ denote the first half of the random input permutation of $M$.
If both $a,b \in F$, or both $a,b \notin F$, then the corresponding payoffs satisfy the required inequality
(as an equality, since $a+b = \frac{a+b}{2} + \frac{a+b}{2}$). 
Suppose now that exactly one element of $\{a,b\}$ appears in $F$, and let $x$ denote the sum of the
other $m-1$ elements in $F$. Note that $a$ precedes $b$ with probability $1/2$ and vice-versa,
so it remains to check that
\begin{equation} \label{eq:inequality}
  \frac12 g(x+a) + \frac12g(x+b) \geq g\left(x + \frac{a+b}{2}\right).
\end{equation}
Note that $-b < - \frac{a+b}{2} < -a$ by the assumption. Consider the four possible intervals
where $x$ can lie in:
$x \leq -b$, $-b \leq x \leq -\frac{a+b}{2}$, $-\frac{a+b}{2} \leq x \leq -a$, $-a \leq x$.
In each of these cases, the verification of~\eqref{eq:inequality} is straightforward:

\begin{enumerate} \itemsep 3pt
\item $x \leq -b$: $\frac12 \cdot 0 +  \frac12 \cdot 0 = 0$.
\item $-b \leq x \leq -\frac{a+b}{2}$: $0 + \frac12 (x+b) \geq 0$.
\item $-\frac{a+b}{2} \leq x \leq -a$: $0 + \frac12 (x+b) \geq x + \frac{a+b}{2}$.
\item $-a \leq x$: $\frac12 (x+a) + \frac12 (x+b) = x + \frac{a+b}{2}$,
\end{enumerate}
and the verification is complete.
\end{proof}

If $M$ satisfies $n_+=n_-=n/2$, we invoke Lemma~\ref{lem:replace} with a
pair of distinct elements of $M$ that are either both positive or both negative,
while maintaining the zero-sum property and the average absolute value
$\mu$. It is not hard to see that $f(\cdot)$ is continuous as a
function of a given $n$-element multiset. Therefore, repeated
invocation of Lemma~\ref{lem:replace} allows us to average all the
elements of $M$ of same sign by compactness. In other words, we have
$f(M) \geq f(B)$, where $B$ is a `binary' multiset consisting of two elements
$\mu$ and $-\mu$, each with multiplicity $n/2$. Besides, recall from Section~\ref{sec:middle-b} that
$f(B) = \Theta(\mu\sqrt{n})$. 

If $M$ fails to satisfy $n_+=n_-=n/2$, we assume without loss of
generality that $n_-<n_+$ and proceed as follows. On each step, we
pick $a$ and $b$ as the most negative and the least positive elements
of $M$, respectively. We have that $|a|>b$ by the zero-sum property.
One application of Lemma~\ref{lem:replace} 
increases $n_-$ by one, and reduces $n_+$ by one while maintaining the
zero-sum property. Note, however, that this operation decreases the
average absolute value $\mu$. After $(n_+-n_-)/2$ such steps, we get a
multiset $M'$ that satisfies $n_+=n_-=n/2$. Since
\[ n_+-\frac{n_+-n_-}{2}  > \frac{n_+}{2} \]
of the largest positive elements of the
initial multiset $M$ remain in $M'$, its average absolute
value $\mu'$ satisfies $\mu'>\mu/2$ by the zero-sum property. Moreover,
the argument from the previous paragraph yields that $f(M') =
\Omega(\mu'\sqrt{n})$. Therefore, we conclude that $f(M) \ge f(M') =
\Omega(\mu\sqrt{n})$, as desired.

\subparagraph{Example.} Let $M=\{-5,-3,-3,+1,+1,+1,+2,+6\}$, with
$n=8$ and $\mu=11/4$. In the first step, we average $-5$ with $+1$ to
balance the number of positive and negative elements. Then we
`average' the elements of the same sign in arbitrary order. The
elements of the relevant action pairs are underlined.  
\begin{align*}
  f(\{\und{-5},-3,-3,\und{1},1,1,2,6\}) &\geq f(\{-3,-3,-2,-2,1,1,\und{2},\und{6}\}) \\
  &\geq f(\{-3,-3,-2,-2,1,1,4,4\}) \\
  &\vdots \\
  &\geq f(\{-2.5,-2.5,-2.5,-2.5,2.5,2.5,2.5,2.5\})\\
  &= 2.5 \cdot f(\{-1,-1,-1,-1,1,1,1,1\}) \\
  &= \frac{5 \cdot 18}{2 \cdot 35}= \frac97.
\end{align*}

\section{Concluding remarks} \label{sec:remarks}

Zero-sum sequences are relevant in the study of dynamic or random processes and/or 
optimization under uncertainty. We formulate two questions for further investigation.
The first one is more technical but its solution may offer additional insight on the problem.
The second one appears to be interesting and motivating for real applications, such as
simulation or trading.

\begin{enumerate}  \itemsep 3pt
 
\item Can the entries of $T$ describing the payoffs of an optimal strategy
  for zero-sum sequences over $\Sigma=\{-1,+1\}$ be derived in closed form? 

\item Can a better expected payoff be guaranteed in the general case, by taking advantage
  of the specificity of the input?

\end{enumerate}

\newpage
\appendix  

\section{Proof of Lemma~\ref{lem:W3}} \label{sec:app}

Our argument relies on a standard Vandermonde identity, see~\cite[Section~5.2]{GKP90},
\begin{equation*}
	\sum_{i=0}^{r}\binom{s}{i}\binom{t}{r-i}=\binom{s+t}{r},
\end{equation*}
and a handful of simple algebraic manipulations:
\begin{align*}
	\sum_{k=1}^m &2k {2m \choose m-k}{2m \choose m+k} = m\sum_{k=0}^m 2{2m \choose m-k}{2m \choose m+k} -2\sum_{k=0}^{m-1} (m-k) {2m \choose m-k}{2m \choose m+k} \\
	&= m{2m \choose m}^2+m\sum_{k=-m}^m {2m \choose m-k}{2m \choose m+k}-2m\sum_{k=0}^{m-1} 2{2m-1 \choose m-k-1}{2m \choose m+k} \\
	&= m{2m \choose m}^2+m\sum_{i=0}^{2m} {2m \choose i}{2m \choose 2m-i}-2m\sum_{k=0}^{m-1} 2{2m-1 \choose m-k-1}{2m-1 \choose m+k-1}\\
	&\hspace{58.5mm}-2m\sum_{k=0}^{m-1} 2{2m-1 \choose m-k-1}{2m-1 \choose m+k}\\
	&=m{2m \choose m}^2+m{4m \choose 2m}-2m\sum_{k=-m}^{m-1} {2m-1 \choose m-k-1}{2m-1 \choose m+k-1}\\
	&\hspace{10mm}-2m{2m-1 \choose m-1}^2-2m\sum_{k=-m+1}^{m-1}{2m-1 \choose m-k-1}{2m-1 \choose m+k}\\
	&=m{2m \choose m}^2+m{4m \choose 2m}-2m\sum_{i=0}^{2m-1} {2m-1 \choose i}{2m-1 \choose 2m-1-i}\\
	&\hspace{10mm}-2m{2m-1 \choose m-1}^2-2m\sum_{i=0}^{2m-2}{2m-1 \choose i}{2m-1 \choose 2m-2-i}\\
	&=m{2m \choose m}^2-2m{2m-1 \choose m-1}^2 +m{4m \choose 2m} -2m{4m-2 \choose 2m-1} -2m{4m-2 \choose 2m-2} \\
	&=m{2m \choose m}^2-2m{2m-1 \choose m-1}^2 = 2m{2m-1 \choose m-1}^2.
\end{align*}


\begin{thebibliography}{99}

\bibitem{CMRS64}
Yuan Shih Chow, Sigaiti Moriguti, Herbert Robbins, and Stephen M. Samuels,
Optimal selection based on relative rank (the “secretary problem”),
\emph{Israel Journal of Mathematics}
\textbf{2(2)} (1964), 81--90.

\bibitem{Fe89}
Thomas S. Ferguson,
Who solved the secretary problem?
\emph{Statistical Science}
\textbf{4(3)} (1989), 282–296.

\bibitem{GKMN08}
Nicholas Georgiou, Ma{\l}gorzata Kuchta, Micha{\l} Morayne, and Jaros{\l}aw Niemiec,
On a universal best choice algorithm for partially ordered sets,
\emph{Random Structures and Algorithms}
\textbf{32(3)} (2008), 263--273.

\bibitem{GKP90}
Ronald L. Graham, Donald E. Knuth, and Oren Patashnik,
\emph{Concrete Mathematics: A Foundation for Computer Science},
Addison-Wesley, Reading, MA, 1990.

\bibitem{Gu60}
Irwin Guttman,
On a problem of L. Moser,  
\emph{Canadian Mathematical Bulletin}
\textbf{3(1)} (1960), 35--39.

\bibitem{LM85}
  Richard J. Larsen and Morris L. Marx,   
 \emph{An Introduction to Probability and its Applications},
 Prentice--Hall,  Englewood Cliffs, New Jersey, 1985.

\bibitem{MU17}
Michael Mitzenmacher and Eli Upfal,
\emph{Probability and Computing: Randomized Algorithms and Probabilistic Analysis},
2nd~edition, Cambridge University Press, 2017.

\bibitem{Mo56}
  Leo Moser,
  On a problem of Cayley,
\emph{Scripta Mathematica}
\textbf{22} (1956), 289--292.

\bibitem{Mo65}
  Frederick Mosteller,
  \emph{Fifty Challenging Problems in Probability with Solutions},
Addison-Wesley, Reading, MA, 1965.

\bibitem{Mo79}
Sri Gopal Mohanty,
\emph{Lattice Path Counting and Applications},
Academic Press, New York, 1979.
 
\bibitem{Pr99}
James Preater,
The best-choice problem for partially ordered objects,
\emph{Operations Research Letters}
\textbf{25(4)} (1999), 187--190.

\end{thebibliography}
\end{document}